\documentclass[12pt]{amsart}
\usepackage{graphicx}
\usepackage{epsfig,graphics}
\usepackage{amssymb}
\usepackage{latexsym,comment}
\usepackage{amsmath}
\numberwithin{equation}{section}

\newtheorem{theorem}{Theorem}[section]
\newtheorem{remark}[theorem]{Remark}
\newtheorem{lemma}[theorem]{Lemma}
\newtheorem{corollary}[theorem]{Corollary}
\newtheorem{proposition}[theorem]{Proposition}

\title[Persistence in Phage-Bacteria Communities]{Persistence in Phage-Bacteria Communities with Nested and One-to-One Infection Networks}

\author[Daniel A. Korytowski and Hal L. Smith]{}

\begin{document}

\maketitle

\centerline{\scshape Daniel A. Korytowski and Hal L. Smith}
\medskip
{\footnotesize
 \centerline{ School of Mathematical and Statistical Sciences}
 \centerline{Arizona State University}
   \centerline{Tempe, AZ 85287, USA}
}

\begin{abstract}
We show that a bacteria and bacteriophage system with either a perfectly nested or a one-to-one infection network is permanent, a.k.a uniformly persistent, provided that bacteria that are superior competitors for nutrient devote the least to defence against infection and the virus that are the most efficient at infecting
host have the smallest host range.
By ensuring that the density-dependent reduction in bacterial growth rates are independent of bacterial strain, we are able to arrive at the permanence conclusion sought by Jover et al \cite{Jover}.
The same permanence results hold for the one-to-one infection network considered by Thingstad \cite{Th} but without virus efficiency ordering.
Additionally we show the global stability for the nested infection network, and the global dynamics for the one-to-one network.
\end{abstract}

\section{Introduction}

Jover, Cortez, and Weitz \cite{Jover} observe that some bipartite infection networks
in bacteria and virus communities tend to have a nested structure, characterized by a hierarchy
among both host and virus strains, which determines which virus may infect which host. They argue that trade-offs between
competitive ability of the bacteria hosts and defence against infection and, on the part of virus, between virulence and transmissibility
versus host range  can sustain
a nested infection network (NIN).  Specifically, they find that:  ``bacterial growth rate should decrease
with increasing defence against infection'' and  ``the efficiency of viral infection should decrease
with host range''. Their findings are based on the analysis of a  Lotka-Volterra  model incorporating the above-mentioned trade-offs which
strongly suggests that the perfectly nested community structure
of $n$-host bacteria and $n$-virus is permanent,  or uniformly
persistent \cite{HS,ST}.

Inspired by their work, in \cite{KS} we replace the Lotka-Volterra model by
a chemostat based model in which bacteria compete for a growth-limiting nutrient. In a chemostat model,
each bacterial strain is endowed with a break-even concentration, $R^*$, of nutrient below which it cannot grow such that,
in the absence of virus, only the strain with smallest $R^*$ survives. Thus, within a community of bacteria competing for a single limiting nutrient, the competitiveness of the various  strains are naturally ordered
by their $R^*$ values. In \cite{KS}, we show that nested infection networks of arbitrary size are permanent
provided that $R^*$ values increase
with increasing defence against infection and that the efficiency of viral infection should decrease
with host range. We also show how a bacteria-virus community with NIN of arbitrary size can be assembled by
the successive addition of one new species at a time, answering the question of ``How do NIN come to be?".

We show that the Lotka-Volterra based model of Jover et al \cite{Jover} can be modified in such a way that
the permanence conclusions which they sought can be attained. The key is to ensure that density-dependent reduction in bacterial growth rates
be independent of bacterial strain. Following \cite{Jover}, we assume that virus strain $V_i$ is characterized by its adsorption rate $\phi_i$ and its burst size $\beta_i$, both of which are assumed to be independent of which host strain it infects, and its specific death rate $n_i$. The density of bacteria strain $i$ is denoted by $B_i$, and its specific growth rate is $r_i$. The ``mean field'',
density-dependent depression of growth due to inter and intra-specific competition term $\sum_j a_jB_j$ is common to all strains. The equations of our model are the following.

\begin{eqnarray}\label{eqns}
B_i' &=& B_i\left(r_i-\sum\limits_{j=1}^n a_jB_j\right)-B_i \sum_{j=1}^n M_{ij}\phi_jV_j\\
V_i' &=& \beta_i\phi_iV_i \sum_{j=1}^n M_{ji} B_j-n_iV_i,\ 1\le i\le n. \nonumber
\end{eqnarray}
where matrix $M$ captures the infection network structure:
$$
M_{ij}=\left\{
         \begin{array}{cc}
           1, & V_j\ \hbox{infects}\ H_i \\
           0, &  V_j\ \hbox{does not infect}\ H_i \\
         \end{array}
       \right\}
$$
In the system considered in \cite{Jover}, the bacterial host dynamics in the absence of virus is modeled
as $B_i'=r_iB_i(1-K^{-1}\sum_j B_j)$; a consequence of this is the simplex of equilibria $\sum_j B_j=K$. We avoid this.

Motivated by the work of Jover et al \cite{Jover} and the work of Thingstad \cite{Th},
we consider two special network structures: nested infection networks (NIN) with upper triangular matrix $M$, and one-to-one infection networks (OIN) with $M=I$, the identity matrix.

The scaling of variables
$$
P_i=\phi_iV_i, \ H_i=B_i, \ e_i=\frac{\beta_i \phi_i}{n_i},
$$
exposes a virus infection efficiency parameter $e_i$ for each virus. Hereafter, we consider the resulting scaled
system:

\begin{eqnarray}\label{eqns-scaled}
H_i' &=& H_i\left(r_i-\sum\limits_{j=1}^n a_jH_j\right)-H_i \sum_{j=1}^n M_{ij}P_j\\
P_i' &=& e_in_iP_i \left(\sum_{j=1}^n M_{ji}H_j-\frac{1}{e_i}\right),\ 1\le i\le n. \nonumber
\end{eqnarray}

System \eqref{eqns-scaled} defines a dissipative dynamical system on the nonnegative orthant of $\mathbb{R}^{2n}$

\begin{proposition}\label{bound} Solutions of \eqref{eqns-scaled} with nonnegative (positive) initial data are well-defined for all $t\ge 0$ and remain nonnegative (positive).
In addition, the system has a compact global attractor. Indeed, if
 $F(t)=\sum\limits_{i=1}^{n}H_i(t)+\sum\limits_{i=1}^{n}\frac{P_i(t)}{e_in_i}$ then $$F(t)\le \frac{Q}{W}+(F(0)-\frac{Q}{W})e^{-Wt} \le \max\{F(0),\frac{Q}{W}\},$$
 and
 $$\limsup_{t\to\infty}F(t)\le \sum_{i=1}^{n}(1+\frac{r_i}{W})\frac{r_i}{a_i},$$
 where $K=\max\limits_{i=1}^{n}\{H_i(0),\frac{r_i}{a_i}\}$, $W=\min\limits_{i=1}^{n}\{n_i\}$ and $Q = \sum\limits_{i=1}^{n}(W+r_i)K$.
\end{proposition}

\begin{proof}
Existence and positivity of solutions follow from the form of the right hand side.
Therefore, $H_i'(t) \le H_i(t)(r_i-a_iH_i(t))$.  Hence $H_i(t) \le K$ and $\limsup_{t\to\infty}H_i(t)\le r_i/a_i$.
\begin{align*}
&\frac{dF}{dt} = \sum\limits_{i=1}^{n}r_iH_i-(\sum_{i=1}^n H_i)(\sum\limits_{j=1}^{n}a_jH_j)-\sum\limits_{i=1}^{n}\frac{P_i}{e_i} \le \sum\limits_{i=1}^{n}r_iH_i-W\sum\limits_{i=1}^{n}\frac{P_i}{e_in_i}\\
& =\sum\limits_{i=1}^{n}(W+r_i)H_i-W F.
\end{align*}
The estimate on $F(t)$ follows by bounding the first summation by $Q$ and integrating; the estimate on the limit superior
follows from the estimate of the limit superior of the $H_i$ above and by integration.
\end{proof}

\section{Nested Infection Networks}

If $M$ is upper triangular, then our system becomes:

\begin{eqnarray}\label{LV}
H_i' &=& H_i\left(r_i-\sum\limits_{j=1}^n a_jH_j-\sum\limits_{j \ge i}P_j\right)\\
P_i' &=& e_in_iP_i \left(\sum\limits_{j \le i}H_j-\frac{1}{e_i}\right),\ 1\le i\le n. \nonumber
\end{eqnarray}

Let $E^{*}$ be the equilibrium of the system where each component is positive.  From the second equation of (\ref{LV}),\\
\begin{equation*}\label{Hequilibrium}
H_1^{*}=\frac{1}{e_1},\ H_j^{*}=\frac{1}{e_j}-\frac{1}{e_{j-1}}, j>1
\end{equation*}
These are all positive if
\begin{equation}\label{enorder}
e_1>e_2>e_3>\cdots>e_n
\end{equation}

Let $F_i(H)=r_i-\sum\limits_{j=1}^n a_jH_j$, then at $E^{*}$ from the first equation of (\ref{LV}), $F_i(H^{*})=\sum\limits_{j \ge i}P_j^{*}$.  Let
\begin{equation}\label{q}
Q_n=\frac{a_1}{e_1}+(\frac{a_2}{e_2}-\frac{a_2}{e_1})+(\frac{a_3}{e_3}-\frac{a_3}{e_2})+\cdots+(\frac{a_n}{e_n}-\frac{a_n}{e_{n-1}})
\end{equation}
The right hand side of $F_i(H^{*})$ is positive, and decreases as i increases, therefore $F_i(H^{*})$ needs to be positive and decreasing which is satisfied by
\begin{equation}\label{ss}
Q_n<r_n
\end{equation}
and
\begin{equation}\label{r}
r_1>r_2>\cdots>r_n.
\end{equation}

Inequalities \eqref{ss} and \eqref{r} can be combined to give
\begin{equation}\label{combo}
\frac{a_1}{e_1} \le Q_j \le Q_n < r_n < r_j, 1 \le j < n.
\end{equation}

\begin{proposition}\label{positivequil}
$E^{*}$ exists if and only if (\ref{enorder}), (\ref{ss}) and (\ref{r}) hold.

In fact,
\begin{eqnarray}\label{equil}
H_1^*&=& \frac{1}{e_1},\ H_j^*=\frac{1}{e_j}-\frac{1}{e_{j-1}},\ j>1,\\
P_j^*&=&r_j-r_{j+1},\ j<n,\ P_n^*=r_n-Q_n.\nonumber
\end{eqnarray}
Furthermore, the above also implies the existence of a unique equilibrium $E^{\dag}$ with all components positive except for $P_n=0$.  In fact,
\begin{eqnarray}\label{equiln}
H_n^{\dag}&=& H_n^{*}+\frac{P_n^{*}}{a_n},\ H_j^{\dag}=H_j^{*},\ 1 \le j < n,\\
P_j^{\dag}&=&P_j^{*},\ j<n,\ P_n^{\dag}=0.\nonumber
\end{eqnarray}

\end{proposition}

\begin{remark}\label{Qn}
\eqref{combo} implies the existence of an unique family of equilibria $E^*_k$ with $H_j,P_j=0, \ j>k$ described by \eqref{equil}, but with
$Q_k$ replacing $Q_n$. Another family of equilibria, $E^\dag_k$, exists with $H_j=0, \ j>k$ and $P_j=0,\ j\ge k$  described by \eqref{equiln}, but with $Q_k$ replacing $Q_n$.
There are many other equilibria but we have no need to enumerate all of them.
\end{remark}


Hereafter, we assume without further comment that (\ref{ss}), and (\ref{r}) hold.

We write $H_{i,\infty} = \liminf\limits_{t \to \infty}H_i(t)$ and $H_i^{\infty}$ with limit
superior in place of limit inferior.

\begin{proposition}\label{limsup}
\begin{itemize}
  \item[(a)] If $(\sum\limits_{j \le i} H_j(t))^{\infty}<\frac{1}{e_i}$ then $P_i(t) \to 0$.
  \item[(b)] If $i < j, P_i(0) > 0,$ and if $(H_{i+1}+H_{i+2}+\cdots+H_{j})^{\infty}<\frac{1}{e_j}-\frac{1}{e_i}$ then $P_j(t) \to 0$.
  \item[(c)] If $i < j, H_i(0) > 0$, and if $(P_i+P_{i+1}+\cdots+P_{j-1})^{\infty} < (r_i-r_j)$ then $H_j(t) \to 0$.
\end{itemize}
\end{proposition}

\begin{proof} of (a):
The equation for $P_i$ implies that
\begin{equation*}
\frac{d}{dt}\log P_i^{\frac{1}{n_ie_i}}=\sum\limits_{j \le i} H_j(t)-\frac{1}{e_i}
\end{equation*}
If $(\sum\limits_{j \le i}^n H_j(t))^{\infty}<\frac{1}{e_i}$ is false, then $P_i \to \infty$, a contradiction to $P_i$ being bounded.  Assertion (a) is transparent.

Proof of (b): If $i < j, P_i(0) > 0$, and if $(H_{i+1}+H_{i+2}+\cdots+H_{j})^{\infty}<\frac{1}{e_j}-\frac{1}{e_i}$ then
\begin{align*}
\frac{d}{dt}\log \frac{P_i^{\frac{1}{n_ie_i}}}{P_j^{\frac{1}{n_j e_j}}} =& \frac{P_i'}{e_i n_i P_i}-\frac{P_j'}{n_j e_j P_j}\\
=& \frac{-1}{e_i}+\frac{1}{e_j}-(H_{i+1}+H_{i+2}+\cdots+H_{j})\ge \epsilon, t \ge T
\end{align*}
for some $\epsilon, T > 0$.  Therefore, $\frac{P_i^{\frac{1}{n_ie_i}}}{P_j^{\frac{1}{n_j e_j}}} \to \infty$, and since $P_i, P_j$ are bounded, $P_j(t) \to 0$.

Proof of (c): assume that $H_i(0) > 0$ and $(P_i+P_{i+1}+\cdots+P_{j-1})^{\infty} < (r_i-r_j)$, then
\begin{align*}
\frac{d}{dt}\log \frac{H_i(t)}{H_j(t)} =& \frac{H_i'}{H_i}-\frac{H_j'}{H_j}\\
=& (r_i-r_j)-(P_i(t)+P_{i+1}(t)+\cdots+P_{j-1}(t))\\
\end{align*}
It follows that $\frac{H_i}{H_j} \to \infty$, and since $H_i, H_j$ are bounded, $H_j(t) \to 0$.
\end{proof}

If there are no virus present, then host $H_1$ drives the other hosts to extinction.

\begin{lemma}\label{extinction2}
If $P_i \equiv 0, 1 \le i \le n, H_1(0)>0$ then $H_1 \to \frac{r_1}{a_1}$.
\end{lemma}

\begin{proof}
Since $P_i \equiv 0, H_{i+1} \to 0$ by Proposition~\ref{limsup} (c) for $1 \le i < n$.  Therefore $\forall \epsilon > 0, \exists T > 0$ such that $\forall t \ge T, \sum\limits_{j=2}^n a_jH_j(t) < \epsilon$. Then for $t > T, H_1' > H_1(r_1-a_1H_1-2\epsilon)$.  Therefore $H_{1,\infty} \ge \frac{r_1-2\epsilon}{a_1}$ and since $\epsilon>0$ is arbitrary, $H_{1,\infty} \ge \frac{r_1}{a_1}$.
On the other hand, $H_1' \le H_1(r_1-a_1H_1)$, so $H_1^{\infty} \le \frac{r_1}{a_1}$.  Therefore $H_1 \to \frac{r_1}{a_1}$.
\end{proof}

Now we show that $H_1$ persists if initially present regardless of who else is around; similarly, $H_1$ and $V_1$ persist if initially present regardless of which other host and virus are present.

\begin{proposition}\label{weakpersist}
\begin{enumerate}
\item [(a)] If $H_1(0)>0$, then $H_1^\infty\ge \frac{1}{e_1}$.

\item [(b)] If $H_1(0)>0$ and $P_1(0)>0$, then
$$H_1^\infty\ge \frac{1}{e_1}, \ P_1^\infty\ge \min\{r_1-r_2,  \frac{r_1e_1-a_1}{e_1}\}.$$
\end{enumerate}
\end{proposition}

\begin{proof} of (a).  Assume the conclusion is false. Then $P_1\to 0$ by Proposition~\ref{limsup} (a).

If $P_i(0)=0$ for all $i$, then $H_1(t)\to \frac{r_1}{a_1} \ge \frac{1}{e_1}$ by Lemma~\ref{extinction2} and (\ref{combo}), so we suppose that $P_i(0)>0$  for some $i$.
Let $k$ denote the smallest such integer $i$ for which $P_i(0)>0$.

If $k=1$, then, as noted above, $P_1\to 0$ and so $H_2\to 0$ by\\ Proposition~\ref{limsup} (c). Then $P_2\to 0$
by Proposition~\ref{limsup} (a) or (b).

If $k=2$, then $P_1\equiv 0$ so $H_2\to 0$ by Lemma~\ref{limsup} (c) since $H_1$ and $H_2$ share the same virus.
Since $(H_1+H_2)^\infty=H_1^\infty<\frac{1}{e_1}<\frac{1}{e_2}$, it follows that $P_2\to 0$ by Proposition~\ref{limsup} (a). Now we can
use Proposition~\ref{limsup} (c) to show $H_3\to 0$ and then Proposition~\ref{limsup} (a) or (b) to show $P_3\to 0$.

If $k>2$, then $P_1\equiv P_2\equiv \cdots \equiv P_{k-1}\equiv 0$ and $P_k(0)>0$. As $H_1, \cdots, H_{k-1}$ share the same virus, then $H_i\equiv 0$ or $H_i\to 0$ for $1<i\le k-1$ by  Proposition~\ref{limsup} (c).  $H_k\to 0$ by Proposition~\ref{limsup} (c). Then,
$P_k\to 0$ by Proposition~\ref{limsup} (a) since $(\sum\limits_{j\le k}H_j)^\infty=H_1^\infty<\frac{1}{e_1}<\frac{1}{e_k}$.
So $H_{k+1}\to 0$ by Proposition~\ref{limsup} (c). Proposition~\ref{limsup} (a) or (b) implies that
$P_{k+1}\to 0$.

We see that for all values of $k$, $H_2,\cdots, H_{k+1}\to 0$ and $P_1,\cdots, P_{k+1}\to 0$. Successive additional applications of
Proposition~\ref{limsup} (a) or (b) and (c) then imply that $H_2,\cdots, H_{n}\to 0$ and $P_1,\cdots, P_n\to 0$. But, then for all $\epsilon>0$, there exists $T>0$ such that
$$
H_1'/(H_1)\ge r_1-\epsilon-a_1H_1,\ t\ge T.
$$
This implies that $H_1^\infty>\frac{r_1}{a_1}>\frac{1}{e_1}$, by (\ref{ss}), a contradiction.
This completes the proof of the first assertion.\\

Proof of (b): Now, suppose that $H_1(0)>0, P_1(0)>0$ and $P_1^\infty<r_1-r_2$. \\
Proposition~\ref{limsup} (c) implies that $H_2\to 0$. By  Proposition~\ref{limsup} (b), $P_2\to 0$. Applying Proposition~\ref{limsup} (c) with $i=1$ and $j=3$,
as $(P_1+P_2)^\infty=P_1^\infty<r_1-r_2<r_1-r_3$, we conclude that $H_3\to 0$. Then,
Proposition~\ref{limsup} (b) implies that $P_3\to 0$. Clearly, we can continue sequential application of Proposition~\ref{limsup} (b) and (c)
to conclude that $H_i,P_i\to 0$ for $i>1$. Then we use that
\begin{eqnarray*}
  \frac{d}{dt}\log H_1P_1^{\frac{a_1}{e_1 n_1}} &=& \frac{H_1'}{H_1}+\frac{a_1P_1'}{P_1 e_1 n_1} \\
   &=& r_1-\frac{a_1}{e_1}-P_1-\hbox{terms that go to zero}
  \end{eqnarray*}
to conclude that $P_1^\infty\ge \frac{r_1e_1-a_1}{e_1}$.
\end{proof}

\begin{lemma}[Lemma 1.2 \cite{KS}]\label{Hofbauer}
Let $x(t)$ be a bounded positive solution of the Lotka-Volterra system
$$
x_i'=x_i(r_i+\sum\limits_{j=1}^n a_{ij}x_j ),\ 1\le i\le n
$$
and suppose there exists $k<n$ and $m,M,\delta>0$ such that $m\le x_i(t)\le M,\ 1\le i\le k,\ t>0$, $x_{k+1}(t)\le \delta,\ t>0$,
and $x_j(t)\to 0$ for $j>k+1$. Suppose also that the $k\times k$ subsystem obtained by setting $x_j=0, j>k$ has a unique positive
equilibrium $p=(p_1,p_2,\cdots,p_k)$. Then
$$
\liminf_{T\to\infty}\frac{1}{T}\int_0^T x_i(t)dt=p_i+O(\delta),\ 1\le i\le k.
$$
The same expression holds for the limit superior.
\end{lemma}

\begin{theorem}\label{persist}
Let $1\le k\le n$.
\begin{enumerate}
  \item [(a)] There exists $\epsilon_k>0$ such that if $H_i(0)>0,\ 1\le i\le k$ and \\
  $P_j(0)>0,\ 1\le j\le k-1$, then
  $$
  H_{i,\infty}\ge \epsilon_k,\ 1\le i\le k \ \hbox{and}\ P_{j,\infty}\ge \epsilon_k,\ 1\le j\le k-1.
  $$
  \item [(b)] There exists $\epsilon_k>0$ such that if $H_i(0)>0, P_i(0)>0,\ 1\le i\le k$, then
  $$
  H_{i,\infty}\ge \epsilon_k,\ P_{i,\infty}\ge \epsilon_k,\ 1\le i\le k.
  $$
\end{enumerate}
\end{theorem}

\begin{proof} We use the notation $[H_i]_t\equiv \frac{1}{t}\int_0^t H_i(s)ds$.
Our proof is by mathematical induction using the ordering of the $2n$ cases as follows
$$
(a,1)<(b,1)<(a,2)<(b,2)<\cdots<(a,n)<(b,n)
$$
where $(a,k)$ denotes case (a) with index $k$.

The cases $(a,1)$ and $(b,1)$ follow immediately from Proposition~\ref{weakpersist} and Corollary 4.8 in \cite{ST} with persistence function
$\rho=\min\{H_1,P_1\}$ in case $(b,1)$.  The latter result says that weak (limsup) uniform persistence implies strong (liminf) uniform persistence when the dynamical system
is dissipative.

For the induction step, assuming that $(a,k)$ holds, we prove that $(b,k)$ holds and assuming that $(b,k)$ holds, we prove that $(a,k+1)$ holds.

We begin by assuming that $(a,k)$ holds and prove that $(b,k)$ holds. We consider solutions satisfying $H_i(0)>0, P_i(0)>0$ for $1\le i\le k$.
Note that other components $H_j(0)$ or $P_j(0)$ for $j>k$ may be positive or zero, we make no assumptions.
As $(a,k)$ holds, there exists $\epsilon_k>0$ such that $H_{i,\infty}\ge \epsilon_k,\ 1\le i\le k$ and $P_{i,\infty}\ge \epsilon_k,\ 1\le i\le k-1$. We need only show the
existence of $\delta>0$ such that $P_{k,\infty}\ge \delta$ for every solution with initial values as described above. In fact, by Corollary 4.8 in \cite{ST},
weak uniform persistence implies strong uniform persistence, it suffices to show that $P_k^\infty \ge \delta$.

If $P_k^\infty<r_k-r_{k+1}$, then $H_{k+1}\to 0$ by Proposition~\ref{limsup} (c). Then, by Proposition~\ref{limsup} (b), $P_{k+1}\to 0$.
Clearly,  we may sequentially apply Proposition~\ref{limsup} (b) and (c) to show that $H_j\to 0, P_j\to 0$ for $j\ge k+1$.

If there is no $\delta>0$ such that $P_k^\infty \ge \delta$ for every solution with initial data as described above, then for every
$\delta>0$, we may find a solution with such initial data such that $P_k^\infty < \delta$. By a translation of time, we may assume that
$P_k(t)\le \delta,\ t\ge 0$ for $0<\delta<r_k-r_{k+1}$ to be determined later. Then $H_j,P_j\to 0, j\ge k+1$.
Now, as $(a,k)$ holds, we may apply Lemma~\ref{Hofbauer}. The subsystem with $H_i=0, \ k+1\le
i\le n$ and $P_i=0,\ k\le i\le n$ has a unique positive equilibrium
by Proposition~\ref{positivequil}. See Remark~\ref{Qn}. The equation
$$
\frac{P_k'}{P_k e_k n_k}=\sum\limits_{j\le k}H_j-\frac{1}{e_k}
$$
implies that
$$
\frac{1}{t}\log \frac{P_k^{\frac{1}{e_k n_k}}(t)}{P_k^{\frac{1}{e_k n_k}}(0)}=\sum\limits_{j\le k}[H_j]_t -\frac{1}{e_k}.
$$
By \eqref{equiln} and Lemma~\ref{Hofbauer}, we have for large $t$
$$
\sum\limits_{j\le k}[H_j]_t-\frac{1}{e_k}=\sum\limits_{j\le k}H_j^\dag-\frac{1}{e_k}+O(\delta)=\frac{1}{e_{k-1}}+q-\frac{1}{e_k}+O(\delta)
$$
where $q=\frac{1}{a_k}(r_k-Q_{k-1})+\frac{1}{e_{k-1}}-\frac{1}{e_k}>0$. On choosing $\delta$ small enough and an appropriate solution, then $\sum\limits_{j\le k}[H_j]_t -\frac{1}{e_k}>q/2$ for large $t$, implying that $P_k\to +\infty$, a contradiction. We have proved that $(a,k)$ implies $(b,k)$.

Now, we assume that $(b,k)$ holds and prove that $(a,k+1)$ holds. We consider solutions satisfying\\
$H_i(0)>0, P_i(0)>0$ for $1\le i\le k$ and $H_{k+1}(0)>0$.  As $(b,k)$ holds by assumption, and following the same arguments as in the previous case, we only need to show that there exists $\delta>0$ such that $H_{k+1}^\infty\ge \delta$ for all solutions with initial data
as just described.

If $H_{k+1}^\infty<\frac{1}{e_{k+1}}-\frac{1}{e_k}$, then $P_{k+1}\to 0$ by Proposition~\ref{limsup} (b) and then\\
$H_{k+2}\to 0$ by Proposition~\ref{limsup} (c).  This reasoning may be iterated to yield $H_i\to 0,\ k+2\le i\le n$ and $P_i\to 0,\ k+1\le i\le n$.

If there is no $\delta>0$ such that $H_{k+1}^\infty \ge \delta$ for every solution with initial data as described above, then for every
$\delta>0$, we may find a solution with such initial data such that $H_{k+1}^\infty < \delta$. By a translation of time, we may assume that
$H_{k+1}(t)\le \delta,\ t\ge 0$ for $0<\delta<\frac{1}{e_{k+1}}-\frac{1}{e_k}$ to be determined later. Then $H_j,P_j\to 0, j\ge k+2$ and $P_{k+1}\to 0$.
Now, using that $(b,k)$ holds, we apply Lemma~\ref{Hofbauer}. The subsystem with $H_i=0, P_i=0 \ k+1\le
i\le n$ has a unique positive equilibrium
by Proposition~\ref{positivequil}. See Remark~\ref{Qn}.
The equation for $H_{k+1}$ is
$$
\frac{H_{k+1}'}{H_{k+1}}=r_{k+1}-\sum\limits_{j=1}^{k}a_jH_j-\sum\limits_{j=k+1}^na_jH_j-\sum\limits_{j=k+1}^nP_j
$$
Integrating, we have
$$
\frac{1}{t}\log \frac{H_{k+1}(t)}{H_{k+1}(0)}=\sum\limits_{j=1}^{k+1}a_j[H_j]_t+O(1/t)
$$

By Remark \ref{Qn} and Lemma~\ref{Hofbauer}, we have that for all large $t$\\
By \eqref{equil} and Lemma~\ref{Hofbauer}, we have that for all large $t$

$$
\sum\limits_{j=1}^ka_j[H_j]=\sum\limits_{j=1}^ka_jH_j^*+O(\delta)=Q_n+O(\delta)
$$
since $H_{k+1}(t)\le\delta$, $[H_{k+1}]_t=O(\delta)$.
Now, $Q_n>0$ so by choosing $\delta$ sufficiently small and an appropriate solution, we can ensure that
the right hand side is bounded below by a positive constant for all large $t$, implying that $H_{k+1}(t)$ is unbounded.
This contradiction completes our proof that $(b,k)$ implies $(a,k+1)$. Thus, our proof is complete by mathematical induction.
\end{proof}

\begin{corollary}\label{meanvalue}
For every solution of \eqref{LV} starting with all components positive, we have that
\begin{equation}
\frac{1}{t}\int_0^t H_i(s)ds \to H_i^*, \ \frac{1}{t}\int_0^t P_i(s)ds\to P_i^*
\end{equation}
where $H_i^*, P_i^*$ are as in \eqref{equil}.

For every solution of \eqref{LV} starting with all components positive except\\
$P_n(0)=0$, we have that
\begin{equation}
\frac{1}{t}\int_0^t H_i(s)ds \to H_i^\dag, \ \frac{1}{t}\int_0^t P_i(s)ds\to P_i^\dag
\end{equation}
where $H_i^\dag, P_i^\dag$ are as in \eqref{equiln}.
\end{corollary}

\begin{proof}
This follows from the previous theorem together with Theorem 5.2.3 in \cite{HS}.
\end{proof}

\section{Global Stability for  NIN in a special case}

Positive equilibrium $E^*$ exists so we can write the system as

\begin{eqnarray}\label{LV1}
H_i' &=& H_i\left(\sum\limits\limits_{j=1}^n a_j(H_j^*-H_j)+\sum\limits\limits_{j \ge i}(P_j^*-P_j)\right)\\
P_i' &=& e_in_iP_i \left(\sum\limits\limits_{j \le i}(H_j-H_j^*)\right),\ 1\le i\le n. \nonumber
\end{eqnarray}

Let $U(x,x^*)=x-x^*-x^*\log x/x^*,\ x,x^*>0$, be the familiar Volterra function and let
$$
V=\sum\limits_i c_iU(H_i,H_i^*)+\sum\limits_i d_iU(P_i,P_i^*)
$$
where $c_1,\cdots, c_n$ and $d_1,\cdots,d_n$ are to be determined.

Then the derivative of $V$ along solutions of \eqref{LV1}, $\dot V$, is given by
\begin{eqnarray*}
  \dot V &=& -\left(\sum\limits_i c_i(H_i-H_i^*)\right)\left(\sum\limits_j a_j(H_j-H_j^*)\right)-\sum\limits_i c_i(H_i-H_i^*)\sum\limits_{j\ge i}(P_j-P_j^*) \\\nonumber
   & & +\sum\limits_i d_ie_in_i(P_i-P_i^*)\sum\limits_{j\le i}(H_j-H_j^*)
\end{eqnarray*}
We aim to choose parameters so that the last two terms cancel each other.
The second summation may be rewritten as $\sum\limits_i (P_i-P_i^*)\sum\limits_{j\le i} c_j(H_j-H_j^*)$ so that the last two sums may be combined as
$\sum\limits_i (P_i-P_i^*)\sum\limits_{j\le i} (d_ie_in_i-c_j)(H_j-H_j^*)$. It vanishes if $\forall i, d_ie_in_i-c_j=0,\ j\le i$. Taking $i=n$, we see that
the $c_j$ must be identical so $c_j=a$ for all $j$ for some $a>0$ and $d_i=a/e_in_i$. Therefore, in this case, we have
$$
\dot V =-\left(\sum\limits_i a(H_i-H_i^*)\right)\left(\sum\limits_j a_j(H_j-H_j^*)\right)
$$
If, in addition, $a_j=a$ for all $j$, then we have
\begin{equation}\label{Vdot}
\dot V =-a^2\left(\sum\limits_i H_i-\sum\limits_i H_i^*\right)^2
\end{equation}

\begin{theorem}
Assume that \eqref{ss} holds and  $a_j=a>0, \ 1\le j\le n$ in \eqref{LV1}. Then $E^*$ is globally asymptotically stable relative to the open positive orthant of $\mathbb{R}^{2n}_+$.\\
With the same assumptions, but $H_i^\dag, V_i^\dag$ replacing $H_i^*,V_i^*$ and additionally $P_n \equiv 0$, $E^\dag$ is globally asymptotically stable relative to the open positive orthant of $\mathbb{R}^{2n-1}_+$.
\end{theorem}

\begin{proof} We first note that since $V(H(t),P(t))\le V(H(0),V(0)), t\ge 0$ for every positive solution of \eqref{LV1}, each component is bounded above and below:
$0<p\le x(t)\le P, t\ge 0$, where $x=H_i,P_j$ and $p,P$ may depend on the solution.

Consider a positive solution of \eqref{LV1}. By LaSalle's invariance principle, every point in its (invariant) limit set $L$ must satisfy $\sum\limits_iH_i=\sum\limits_iH_i^*$ since $L\subset \{(H,V):\dot V=0\}$.
Since $V(x)\le V(H(0),P(0))$ for all $x\in L$, $L$ belongs to the interior of the positive orthant and it is bounded away (but maybe not uniformly) from the boundary of the orthant.
We now consider a trajectory belonging to $L$; until further notice, all considerations involve this solution.
Notice that this solution satisfies
\begin{equation}\label{Heqn}
H_i'= H_i\left(\sum\limits\limits_{j \ge i}(P_j^*-P_j)\right)
\end{equation}
\begin{equation}\label{Peqn}
P_i' = e_in_iP_i \left(\sum\limits\limits_{j \le i}(H_j-H_j^*)\right)
\end{equation}

From \eqref{Peqn}, we see that $P_n'\equiv 0$ so $P_n(t)$ is constant.
Then, $H_n'=H_n(P_n^*-P_n)$ so $H_n(t)$ is either converging exponentially fast to zero, blowing up to infinity, or identically constant depending on the value of
$P_n$. The only alternative that is consistent with $L$ being invariant, bounded, and bounded away from the boundary of the orthant is that $H_n(t)$ is constant
and that $P_n=P_n^*$. As we use a similar argument repeatedly below, we refer to it as our standard argument.

Since $H_n$ is constant and $\sum_i H_i$ is constant, equal to $\sum_i H_i^*$, then so is $\sum_{i\le n-1}H_i$ a constant. But now we face the same dilemma as above with the equation \eqref{Peqn} with $i=n-1$ since the sum in parentheses is constant. By our standard argument, the only alternative is that this constant is zero, i.e., that
$\sum_{i\le n-1}H_i=\sum_{i\le n-1}H_i^*$ and $P_{n-1}(t)$ is constant.  The former implies that
$$
H_n=\sum_i H_i-\sum_{i\le n-1}H_i=H_n^*.
$$

Suppose that $1<k\le n$ and that $H_i(t)\equiv H_i^*,\ P_i(t)\equiv P_i^*,\ i\ge k$, hold. We claim that $H_{k-1}(t)\equiv H_{k-1}^*,\ P_{k-1}(t)\equiv P_{k-1}^*$.
As $P_k(t)$ is constant, \eqref{Peqn} implies that $\sum_{j\le k}H_j(t)=\sum_{j\le k}H_j^*$ and since $H_k=H_k^*$, it follows that $\sum_{j\le k-1}H_j(t)=\sum_{j\le k-1}H_j^*$.
Notice that if $k=2$, then the latter gives that $H_1=H_1^*$.
Now from \eqref{Peqn}, $P_{k-1}'(t)=0$ so $P_{k-1}(t)$ is constant. This implies, by \eqref{Heqn} and our standard argument, that $H_{k-1}'=0$ and $P_{k-1}(t)=P_{k-1}^*$.
If $k=2$, we are done: $H_1=H_1^*,\ P_1=P_1^*$. If $k>2$, then $\sum_{j\le k-2}H_j(t)=\sum_{j\le k-1}H_j(t)-H_{k-1}(t)$ is constant so from \eqref{Peqn} and our standard argument
we conclude that $P_{k-2}'=0$ and that $\sum_{j\le k-2}H_j(t)=\sum_{j\le k-2}H_j^*$. The latter implies that
$$
H_{k-1}=\sum_{j\le k-1}H_j-\sum_{j\le k-2}H_j=\sum_{j\le k-1}H_j^*-\sum_{j\le k-2}H_j^*=H_{k-1}^*.
$$
This completes our proof of the claim. By induction, we conclude that $H_i(t)\equiv H_i^*,\ P_i(t)\equiv P_i^*,\ 1\le i\le n$, i.e., our solution is identical to $E^*$.
Since we considered an arbitrary solution
starting at a point of $L$, it follows that $L=\{E^*\}$. As our chosen solution was an arbitrary positive solution, we have established the result.\\
The arguments are nearly identical for the $E^\dag$ case.  From $\eqref{Heqn}$, $H_n'=0$ since $P_n \equiv 0$, therefore the standard argument starts at $n-1$ instead.
\end{proof}

\section{One-to-One Infection Network}

 $M=I$ in the one-to-one infection network so the equations then becomes:

\begin{eqnarray}\label{LVmono}
H_i' &=& H_i\left(r_i-\sum\limits_{j=1}^n a_jH_j\right)-H_iP_i\\
P_i' &=& e_in_iP_i\left(H_i-\frac{1}{e_i}\right),\ 1\le i\le n. \nonumber
\end{eqnarray}

The principle equilibrium for the one-to-one infection network are now described.
\begin{proposition}\label{positivequilmono}
There exists an equilibrium $E^*$ with $H_i$ and $P_i$ positive for all $i$ if and only if the following inequality holds:
\begin{equation}\label{ssmono}
\quad \tilde Q_n<r_j, 1 \le j \le n, \quad \tilde Q_n=\sum\limits_{i=1}^n \frac{a_i}{e_i}.
\end{equation}
In fact,
\begin{eqnarray}\label{equilmono}
H_j^*&=& \frac{1}{e_j},\ j\ge 1,\\
P_j^*&=&r_j-\tilde Q_n,\ j\ge 1.\nonumber
\end{eqnarray}
The positive equilibrium $E^*$ is unique.

We also note the existence of a unique equilibrium $E^\dag$, with all components positive except for $P_n=0$,
given by
\begin{eqnarray}\label{equilnmono}
H_j^\dag&=&H_j^*,\ 1\le j<n,\nonumber \\
H_n^\dag&=&H_n^{*}+\frac{P_n^{*}}{a_n},\\
P_j^\dag&=&P_j^{*}-P_n^{*}=r_j-r_n,\ j \le n,\nonumber
\end{eqnarray}
provided that $r_n<r_j, j\ne n$.
\end{proposition}

\begin{remark}\label{Qnmono}
We will assume hereafter that
\begin{equation}\label{rmono}
r_1>r_2>\cdots >r_n.
\end{equation}
This hypothesis ensures the existence of a family of equilibria $E_k^*, E_k^\dag,\ 1\le k\le n$, characterized as follows.
$E^*_k$ with $H_j,P_j=0, \ j>k$ is described by \eqref{equilmono} but with
$\tilde Q_k$ replacing $\tilde Q_n$.  $E^\dag_k$ satisfies $H_j=0, \ j>k$ and $P_j=0,\ j\ge k$  described by \eqref{equilnmono} but
with $\tilde Q_k$ replacing $\tilde Q_n$.
\end{remark}

\begin{proposition}\label{limsupmono}

\begin{itemize}
  \item[(a)] If $H_i(t)^{\infty}<\frac{1}{e_i}$ then $P_i(t) \to 0$.
  \item[(b)] If $i < j$, $H_i(0) > 0$, and  $(P_i-P_j)^{\infty} < r_i-r_j$, then $H_j(t) \to 0$.
\end{itemize}
\end{proposition}

\begin{proof} of (a).
The equation for $P_i$ implies that\\
\begin{equation*}
\frac{d}{dt}\log P_i^{\frac{1}{n_ie_i}}=H_i(t)-\frac{1}{e_i}
\end{equation*}
If $H_i(t)^{\infty}<\frac{1}{e_i}$ then $P_i \to 0$.\\

Proof of (b).
Assume that $i < j, H_i(0), H_j(0)>0$ and $(P_i-P_j)^{\infty} < r_i-r_j$. As
$$
\frac{d}{dt}\log \frac{H_i(t)}{H_j(t)} = \frac{H_i'}{H_i}-\frac{H_j'}{H_j}
= (r_i-r_j)-(P_i(t)-P_j(t)),
$$
it follows that $\frac{H_i}{H_j} \to \infty$, which by the boundedness of $H_i, H_j$, implies that $H_j(t) \to 0$.
\end{proof}

\begin{lemma}\label{extinction2mono}
If $P_1 \equiv 0, H_1(0)>0$ then $H_1 \to \frac{r_1}{a_1}$.
\end{lemma}

\begin{proof}
Since $P_1 \equiv 0, H_i \to 0$ by Proposition $\ref{limsup}$ (b) for $1 < i \le n$.  Therefore, $\forall \epsilon > 0, \exists T > 0$ such that $\forall t \ge T, \sum\limits_{j=2}^n a_jH_j(t) < \epsilon$. Then for $t > T, H_1' > H_1(r_1-a_1H_1-2\epsilon)$.  Therefore $H_{1,\infty} \ge \frac{r_1-2\epsilon}{a_1}$ and since $\epsilon>0$ is arbitrary, $H_{1,\infty} \ge \frac{r_1}{a_1}$.
On the other hand, $H_1' \le H_1(r_1-a_1H_1)$, so $H_1^{\infty} \le \frac{r_1}{a_1}$.  Therefore $H_1 \to \frac{a_1}{r_1}$.
\end{proof}

\begin{proposition}\label{weakpersistmono}
\begin{enumerate}
\item [(a)] If $H_1(0)>0$, then $H_1^\infty\ge \frac{1}{e_1}$.

\item [(b)] If $H_1(0)>0$ and $P_1(0)>0$, then
$$H_1^\infty\ge \frac{1}{e_1}, \ P_1^\infty\ge \min\{r_1-r_2,  \frac{r_1e_1-a_1}{e_1}\}.$$
\end{enumerate}
\end{proposition}

\begin{proof} of (a): Assume the conclusion is false. Then $P_1 \to 0$ by Proposition~\ref{limsupmono} (a).  Then $H_2 \to 0$ by Proposition~\ref{limsupmono} (b).  Therefore by sequential applications of Proposition~\ref{limsupmono} (a) and (b), we can conclude that $H_i,P_i \to 0$, for $i>1$.
But, then
$$
H_1'/(H_1)\ge r_1-\epsilon-a_1H_1>r_n-\epsilon-a_1H_1>\frac{a_1}{e_1}-\epsilon-a_1H_1,\ t\ge T
$$
for some $\epsilon>0$ and $T>0$ (recall that $r_i > r_n$ from $(\ref{rmono})$ and $\frac{a_1}{e_1}<r_1$ from (\ref{ssmono})). This implies that $H_1^\infty>\frac{1}{e_1}$, a contradiction.  This completes the proof of the first assertion.

Proof of (b): Now, suppose that $H_1(0)>0, P_1(0)>0$ and $P_1^\infty<r_1-r_2$.  Then $(P_1-P_2)^\infty \le P_1^\infty<r_1-r_2$, therefore Proposition~\ref{limsupmono} (b) implies that
$H_2\to 0$.  Then $P_2 \to 0$ by Proposition~\ref{limsupmono} (a).  Then $(P_2-P_3)^\infty \le P_2^\infty<r_2-r_3$ therefore Proposition~\ref{limsupmono} (b) implies that $H_3\to 0$. Clearly, we can continue sequential applications of Proposition~\ref{limsupmono} (a) and (b) to conclude that $H_i,P_i\to 0$ for $i>1$.
\begin{eqnarray*}
  \frac{d}{dt}\log H_1P_1^{\frac{a_1}{e_1 n_1}} &=& \frac{H_1'}{H_1}+\frac{a_1P_1'}{P_1 e_1 n_1} \\
   &=& r_1-\frac{a_1}{e_1}-P_1-\hbox{terms that go to zero}
  \end{eqnarray*}
to conclude that $P_1^\infty\ge \frac{r_1e_1-a_1}{e_1}$.
\end{proof}


\begin{theorem}\label{persistmono}
Let $1\le k\le n$.
\begin{enumerate}
  \item [(a)] There exists $\epsilon_k>0$ such that if $H_i(0)>0,\ 1\le i\le k$ and $P_j(0)>0,\ 1\le j\le k-1$, then
  $$
  H_{i,\infty}\ge \epsilon_k,\ 1\le i\le k \ \hbox{and}\ P_{j,\infty}\ge \epsilon_k,\ 1\le j\le k-1.
  $$
  \item [(b)] There exists $\epsilon_k>0$ such that if $H_i(0)>0, P_i(0)>0,\ 1\le i\le k$, then
  $$
  H_{i,\infty}\ge \epsilon_k,\ P_{i,\infty}\ge \epsilon_k,\ 1\le i\le k.
  $$
\end{enumerate}
\end{theorem}

\begin{proof} We use the notation $[H_i]_t\equiv \frac{1}{t}\int_0^t H_i(s)ds$.
Our proof is by mathematical induction using the ordering of the $2n$ cases as follows
$$
(a,1)<(b,1)<(a,2)<(b,2)<\cdots<(a,n)<(b,n)
$$
where $(a,k)$ denotes case (a) with index $k$.

The cases $(a,1)$ and $(b,1)$ follow immediately from Proposition~\ref{weakpersistmono} and Corollary 4.8 in \cite{ST} with persistence function
$\rho=\min\{H_1,P_1\}$ in case $(b,1)$.

For the induction step, assuming that $(a,k)$ holds, we prove that $(b,k)$ holds and assuming that $(b,k)$ holds, we prove that $(a,k+1)$ holds.

We begin by assuming that $(a,k)$ holds and prove that $(b,k)$ holds. We consider solutions satisfying $H_i(0)>0, P_i(0)>0$ for $1\le i\le k$.
Note that other components $H_j(0)$ or $P_j(0)$ for $j>k$ may be positive or zero, we make no assumptions.
As $(a,k)$ holds, there exists $\epsilon_k>0$ such that $H_{i,\infty}\ge \epsilon_k,\ 1\le i\le k$ and $P_{i,\infty}\ge \epsilon_k,\ 1\le i\le k-1$. We need only show the
existence of $\delta>0$ such that $P_{k,\infty}\ge \delta$ for every solution with initial values as described above. In fact, by the above-mentioned result
that weak uniform persistence implies strong uniform persistence, it suffices to show that $P_k^\infty \ge \delta$.

If $P_k^\infty<r_k-r_{k+1}$, then $H_{k+1}\to 0$ by Proposition~\ref{limsupmono} (b). Then, by\\
Proposition~\ref{limsupmono} (a), $P_{k+1}\to 0$.  Clearly,  we may sequentially apply\\
Proposition~\ref{limsupmono} (b) and (a) to show that $H_j\to 0, P_j\to 0$ for $j\ge k+1$.

If there is no $\delta>0$ such that $P_k^\infty \ge \delta$ for every solution with initial data as described above, then for every
$\delta>0$, we may find a solution with initial data such that $P_k^\infty < \delta$. By a translation of time, we may assume that
$P_k(t)\le \delta,\ t\ge 0$ for $0<\delta<r_k-r_{k+1}$ to be determined later. Then $H_j,P_j\to 0, j\ge k+1$.
Now, as $(a,k)$ holds, we may apply Lemma~\ref{Hofbauer}. The subsystem with $H_i=0, \ k+1\le
i\le n$ and $P_i=0,\ k\le i\le n$ has a unique positive equilibrium
by Proposition~\ref{positivequilmono}. See Remark~\ref{Qnmono}. The equation
$$
\frac{P_k'}{P_k e_k n_k}=H_k-\frac{1}{e_k}
$$
implies that
$$
\frac{1}{t}\log \frac{P_k^{\frac{1}{e_k n_k}}(t)}{P_k^{\frac{1}{e_k n_k}}(0)}=[H_k]_t -\frac{1}{e_k}.
$$
By \eqref{equilnmono} and Lemma~\ref{Hofbauer}, we have for large $t$
$$
[H_k]_t-\frac{1}{e_k}=H_k^\dag-\frac{1}{e_k}+O(\delta)=H_k^{*}+\frac{P_k^{*}}{a_k}-\frac{1}{e_k}+O(\delta)=\frac{P_k^{*}}{a_k}+O(\delta)>0
$$
Implying that $P_k\to +\infty$, a contradiction. We have proved that $(a,k)$ implies $(b,k)$.

Now, we assume that $(b,k)$ holds and prove that $(a,k+1)$ holds. We consider solutions satisfying\\
$H_i(0)>0, P_i(0)>0$ for $1\le i\le k$ and $H_{k+1}(0)>0$.
As $(b,k)$ holds by assumption, and following the same arguments as in the previous case, we only need to show that there exists $\delta>0$ such that $H_{k+1}^\infty\ge \delta$ for all solutions with initial data
as just described.

If $H_{k+1}^\infty<\frac{1}{e_{k+1}}$, then $P_{k+1}\to 0$ by Proposition~\ref{limsupmono} (a) and then\\
$H_{k+2}\to 0$ by Proposition~\ref{limsupmono} (b).  This reasoning may be iterated to yield $H_i\to 0,\ k+2\le i\le n$ and $P_i\to 0,\ k+1\le i\le n$.

If there is no $\delta>0$ such that $H_{k+1}^\infty \ge \delta$ for every solution with initial data as described above, then for every
$\delta>0$, we may find a solution with such initial data such that $H_{k+1}^\infty < \delta$. By a translation of time, we may assume that
$H_{k+1}(t)\le \delta,\ t\ge 0$ for $0<\delta<\frac{1}{e_{k+1}}$ to be determined later. Then $H_j,P_j\to 0, j\ge k+2$ and $P_{k+1}\to 0$.
Now, using that $(b,k)$ holds, we apply Lemma~\ref{Hofbauer}. The subsystem with $H_i=0, P_i=0 \ k+1\le
i\le n$ has a unique positive equilibrium
by Proposition~\ref{positivequilmono}. See Remark~\ref{Qnmono}.
The equation for $H_{k+1}$ is
$$
\frac{H_{k+1}'}{H_{k+1}}=r_{k+1}-\sum\limits_{j=1}^{k}a_jH_j-\sum\limits_{j=k+1}^na_jH_j-P_{k+1}
$$
Integrating, we have
$$
\frac{1}{t}\log \frac{H_{k+1}(t)}{H_{k+1}(0)}=\sum\limits_{j=1}^{k+1}a_j[H_j]_t+O(1/t)
$$
By \eqref{equilmono} and Lemma~\ref{Hofbauer}, we have that for all large $t$
$$
\sum\limits_{j=1}^ka_j[H_j]=\sum\limits_{j=1}^ka_jH_j^*+O(\delta)=\tilde Q_n+O(\delta).
$$
Since $H_{k+1}(t)\le\delta$, $[H_{k+1}]_t=O(\delta)$.
Now, $\tilde Q_n>0$ so by choosing $\delta$ sufficiently small and an appropriate solution, we can ensure that
the right hand side is bounded below by a positive constant for all large $t$, implying that $H_{k+1}(t)$ is unbounded.
This contradiction completes our proof that $(b,k)$ implies $(a,k+1)$. Thus, our proof is complete by mathematical induction.
\end{proof}

\begin{corollary}\label{meanvaluemono}
For every solution of \eqref{LVmono} starting with all components positive, we have that
\begin{equation}
\frac{1}{t}\int_0^t H_i(s)ds \to H_i^*, \ \frac{1}{t}\int_0^t P_i(s)ds\to P_i^*
\end{equation}
where $H_i^*, P_i^*$ are as in \eqref{equilmono}.

For every solution of \eqref{LVmono} starting with all components positive except $P_n(0)=0$, we have that
\begin{equation}
\frac{1}{t}\int_0^t H_i(s)ds \to H_i^\dag, \ \frac{1}{t}\int_0^t P_i(s)ds\to P_i^\dag
\end{equation}
where $H_i^\dag, P_i^\dag$ are as in \eqref{equilnmono}.
\end{corollary}

\begin{proof}
This follows from the previous theorem together with Theorem 5.2.3 in \cite{HS}.
\end{proof}

\section{Global dynamics for the one-to-one network}
Using the positive equilibrium $E^*$, we can write the system as

\begin{eqnarray}\label{MonoLV1}
H_i' &=& H_i\left(\sum\limits_{j=1}^n a_j(H_j^*-H_j)+P_i^*-P_i\right)\\
P_i' &=& e_in_iP_i \left(H_i-H_i^*\right),\ 1\le i\le n. \nonumber
\end{eqnarray}

As before, we define
$$
V=\sum_i c_iU(H_i,H_i^*)+\sum_i d_iU(P_i,P_i^*)
$$
where $c_1,\cdots, c_n$ and $d_1,\cdots,d_n$ are to be determined.

Then the derivative of $V$ along solutions of \eqref{MonoLV1}, $\dot V$, is given by
\begin{eqnarray*}
  \dot V &=& -\left(\sum_i c_i(H_i-H_i^*)\right)\left(\sum_j a_j(H_j-H_j^*)\right)-\sum_i c_i(H_i-H_i^*)(P_i-P_i^*) \\\nonumber
   & & +\sum_i d_ie_in_i(P_i-P_i^*)(H_i-H_i^*)
\end{eqnarray*}
Letting $c_i=a_i$ and $d_i=\frac{a_i}{e_in_i}$ causes the last two summations to cancel each other out.  Therefore in this case we have

\begin{equation}\label{MonoVdot}
\dot V =-\left(\sum_i a_iH_i-\sum_i a_iH_i^*\right)^2
\end{equation}

Below, we use the notation $(H(t),P(t))$ for the $2n$-vector solution $(H_1(t),H_2(t),\cdots,H_n(t),P_1(t),\cdots, P_n(t))$.

\begin{theorem}
The  $\omega$-limit set of a positive solution of \eqref{MonoLV1} is either $E^*$ or it consists of non-constant entire orbits, $(H(t),P(t))$,
satisfying all of the following:
\begin{enumerate}
  \item [(a)] $\sum_{i=1}^{n} a_i H_i(t)=\sum_{i=1}^{n} a_i H_i^*, \ t\in \mathbb{R}$.
  \item [(b)] $\prod_{i=1}^n P_i(t)^{a_i/e_in_i}$ is independent of $t$.
   \item [(c)] $\forall i$, $(H_i(t),P_i(t))$ is a positive solution of the conservative planar system
   \begin{eqnarray}\label{planar}
H_i'&=& H_i\left(P_i^*-P_i\right)\\\nonumber
P_i'& =& e_in_iP_i \left(H_i-H_i^*\right).
\end{eqnarray}
\end{enumerate}

Similarly, the  $\omega$-limit set of a solution of \eqref{MonoLV1} with all components positive except $P_n\equiv 0$ is either $E^\dag$ or it consists of non-constant entire orbits, as in the previous case, on the hyperplane $\sum_{i=1}^{n} a_i H_i(t)=\sum_{i=1}^{n} a_i H_i^\dag$ with $H_n(t)\equiv H_n^\dag$ and with $\prod_{i=1}^{n-1} P_i(t)^{a_i/e_in_i}$ independent of $t$. Furthermore, for
$1\le i<n,\ (H_i(t),P_i(t))$ satisfies \eqref{planar} but with $H_i^\dag, V_i^\dag$ replacing
$H_i^*,V_i^*$.
\end{theorem}

\begin{proof} We first note that since $V(H(t),P(t))\le V(H(0),V(0)), t\ge 0$ for every positive solution of \eqref{MonoLV1}, each component is bounded above and below:
$0<p\le x(t)\le P, t\ge 0$, where $x=H_i,P_j$ and $p,P$ may depend on the solution.

Consider a positive solution of \eqref{MonoLV1}. By LaSalle's invariance principle, every point in its (invariant) limit set $L$ must satisfy $\sum_ia_iH_i=\sum_ia_iH_i^*$ since $L\subset \{(H,V):\dot V=0\}$.  As in the $NIN$ case, $L$ belongs to the interior of the positive orthant and it is bounded away from the boundary of the orthant.
We now consider a trajectory belonging to $L$; until further notice, all considerations involve this solution.
Notice that this solution satisfies \eqref{planar}.
Thus on $L$, the system decouples into $n$ independent
planar conservative systems, the positive solution of which is either periodic or is the positive equilibrium. See e.g. section 2.3 of \cite{HS}.
Notice that $\sum_i \frac{a_iP_i'}{e_in_iP_i}=\sum_i a_i(H_i-H_i^*)=0$, consequently $\prod_{i=1}^n P_i(t)^{a_i/e_in_i}$ is independent of $t$.

If $E^* \in \omega$-limit set, then $E^*=\omega$-limit set, since $E^*$ is stable. Consequently, if $E^*\notin L$, then at least one of the $(H_i,P_i)$ must be a non-trivial periodic orbit.

The arguments are nearly identical for the case that the solution satisfies $P_n\equiv 0$ and other coordinates positive.
Liapunov function $V$ differs from the previous one only in that the sum goes from one to $n-1$ in the second summation and $H_i^\dag, V_i^\dag$ replace
$H_i^*,V_i^*$; the
choice of the $c_i$ and $d_i$ are as before. \eqref{MonoVdot} is changed only in that superscript $\dag$ replaces $*$.

We only note that
the counterpart to \eqref{planar} for $i=n$ reads $H_n'=0$. As $\sum_i a_i(H_i-H_i^\dag)=0$
on the limit set and since  any positive periodic limiting solution must satisfy $\int_0^T H_i dt=H_i^\dag$, it follows that $H_n\equiv H_n^\dag$.

\end{proof}

In the special case that $n=2$, since $H_1$ ($P_1$) can be expressed in terms of $H_2$ ($P_2$), on $\{(H,V):\dot V=0\}$, every solution
in $L$ is periodic (possible constant).

\end{document}